\documentclass[envcountsect,envcountsame,runningheads,a4paper]{llncs}

\usepackage{grd}
\usepackage{cite}
\usepackage{url}
\usepackage{lmodern}
\usepackage{xcolor}

\hypersetup{
    colorlinks,
    linkcolor={red!50!black},
    citecolor={blue!50!black},
    urlcolor={blue!80!black}
}

\makeatletter
\@addtoreset{equation}{theorem}
\@addtoreset{equation}{definition}
\makeatother

\allowdisplaybreaks

\begin{document}
\title{Automating Verification of State Machines \\ with Reactive Designs and Isabelle/UTP}

\titlerunning{Automating Verification of State Machines}

\author{Simon Foster$^{\text{\href{https://orcid.org/0000-0002-9889-9514}{ORCiD}}}$ \and James Baxter \and Ana Cavalcanti \and \\ Alvaro Miyazawa \and Jim Woodcock}

\authorrunning{Simon Foster et al.}

\institute{University of York}

\maketitle

\begin{abstract}
  State-machine based notations are ubiquitous in the description of component systems, particularly in the robotic
  domain. To ensure these systems are safe and predictable, formal verification techniques are important, and can be
  cost-effective if they are both automated and scalable. In this paper, we present a verification approach for a
  diagrammatic state machine language that utilises theorem proving and a denotational semantics based on Unifying
  Theories of Programming (UTP). We provide the necessary theory to underpin state machines (including induction
  theorems for iterative processes), mechanise an action language for states and transitions, and use these to formalise
  the semantics. We then describe the verification approach, which supports infinite state systems, and exemplify it
  with a fully automated deadlock-freedom check. The work has been mechanised in our proof tool, Isabelle/UTP, and so
  also illustrates the use of UTP to build practical verification tools.
\end{abstract}


\section{Introduction}
\label{sec:intro}
The recent drive for adoption of autonomous robots into situations where they interact closely with humans means that
such systems have become safety critical. To ensure that they are both predictable and safe within their applied
context, it is important to adequately prototype them in a variety of scenarios. Whilst physical prototyping is
valuable, there is a limit to the breadth of scenarios that can be considered. Thus, techniques that allow virtual
prototyping, based on mathematically principled models, can greatly enhance the engineering process. In particular,
formal verification techniques like model checking and theorem proving can enable exhaustive coverage of the state
space.

Diagrammatic notations are widely applied in component modelling, and particularly the modelling of robotic controllers
via state machines. Standards like UML\footnote{Unified Modelling Language. \url{http://www.uml.org/}} and
SysML\footnote{Systems Modelling Language. \url{http://www.omgsysml.org/}} provide languages for description of
component interfaces, the system architecture, and the behaviour of individual components. These notations have proved
popular due to a combination of accessibility and precise modelling techniques. In order to leverage formal verification
in this context, there is a need for formal semantics and automated tools. Since UML is highly extensible, a specific
challenge is to provide scalable semantic models that support extensions like real-time, hybrid computation, and
probability.

RoboChart~\cite{Miyazawa2017,Miyazawa2018RoboChart} is a diagrammatic language for the description of robotic
controllers with a formal semantics based on Hoare and He's \emph{Unifying Theories of Programming}~\cite{Hoare&98}
(UTP). The core of RoboChart is a formalised state machine notation that can be considered a subset of UML/SysML state
machine diagrams enriched with time and probability constructs. Each state machine has a well defined interface
describing the events that are externally visible. The behaviour of states and transitions is described using a formal
action language that corresponds to a subset of the \Circus modelling language~\cite{Oliveira&09}. The notation supports
real-time constraints, through delays, timeouts and deadlines, and also probabilistic choices, to express
uncertainty. The use of UTP, crucially, enables us to provide various semantic models for state machines that account
for different computational paradigms, and yet are linked through a common foundation.

In previous work~\cite{Miyazawa2017}, model checking facilities for RoboChart have been developed and applied in
verification. This provides a valuable automated technique for model development, which allows detection of problems
during the early development stages. However, explicit state model checking is limited to checking finite state
models. In practice this means that data types must be abstracted with a small number of elements. In order to
exhaustively check the potentially very large or infinite state space of many robotic applications, symbolic techniques,
like theorem proving, are required. For theorem proving to be practically applicable, like model checking, automation is
highly desirable.



In this paper we present an automated verification technique for a subset of RoboChart state machines in
Isabelle/HOL~\cite{Isabelle}. With it, state machines can be verified against properties formalised in a refinement
statement, such as deadlock freedom. We mechanise the state machine meta-model, including its data types,
well-formedness constraints, and validation support. We use a UTP theory of reactive designs~\cite{Foster17c,Foster18a}
to provide a dynamic semantics, based guarded iteration~\cite{Dijkstra75}. We also engineer automated proof support in
our UTP implementation, Isabelle/UTP~\cite{Foster16a}. The semantics can, therefore, be used to perform verification of
infinite-state systems by theorem proving, with the help of a verified induction theorem, and Isabelle/HOL's automated
proof facilities~\cite{Blanchette2011}. Our denotational approach, like UML, is extensible, and further mechanised UTP
theories can account for real-time~\cite{Sherif2010}, probability~\cite{Bresciani2012}, and other
paradigms~\cite{Foster16b}. Our work also serves as a template for building verification tools with Isabelle/UTP.

In \S\ref{sec:prelim} we outline background material for our work. In \S\ref{sec:foundations} we begin our contributions
by extending reactive designs with guarded iteration and an induction theorem for proving invariants. In
\S\ref{sec:actions} we mechanise reactive programs in Isabelle/UTP, based on the reactive-design theory, and provide
symbolic evaluation theorems. In \S\ref{sec:statsem} we mechanise a static semantics of state machines. In
\S\ref{sec:dynsem} we provide the dynamic semantics, utilising the result from \S\ref{sec:foundations}, and prove a
specialised induction law. In \S\ref{sec:verify} we outline our verification technique, show how to automatically prove
deadlock freedom for an example state machine. Finally, in \S\ref{sec:concl} we conclude and highlight related work.

\section{Preliminaries}
\label{sec:prelim}
\subsection{RoboChart}

RoboChart~\cite{Miyazawa2017} describes robotic systems in terms of a number of controllers that communicate using
shared channels. Each controller has a well defined interface, and its behaviour is described by one or more state
machines. A machine has local state variables and constants, and consists of nodes and transitions, with behaviour
specified using a formal action language~\cite{Oliveira&09}. Advanced features such as hierarchy, shared variables,
real-time constraints, and probability are supported.

A machine, \textsf{GasAnalysis}, is shown in Figure~\ref{fig:gasanalysis}; we use it as a running example. It models a
component of a chemical detector robot~\cite{Timmis2012} that searches for dangerous chemicals using its spectrometer
device, and drops flags at such locations. \textsf{GasAnalysis} is the component that decides how to respond to a sensor
reading. If gas is detected, then an analysis is performed to see whether the gas is above or below a given
threshold. If it is below, then the robot attempts to triangulate a position for the source location and turns toward
it, and if it is above, it stops.

\begin{figure}
  \vspace{-4ex}
  \begin{center}
    \includegraphics[width=\linewidth]{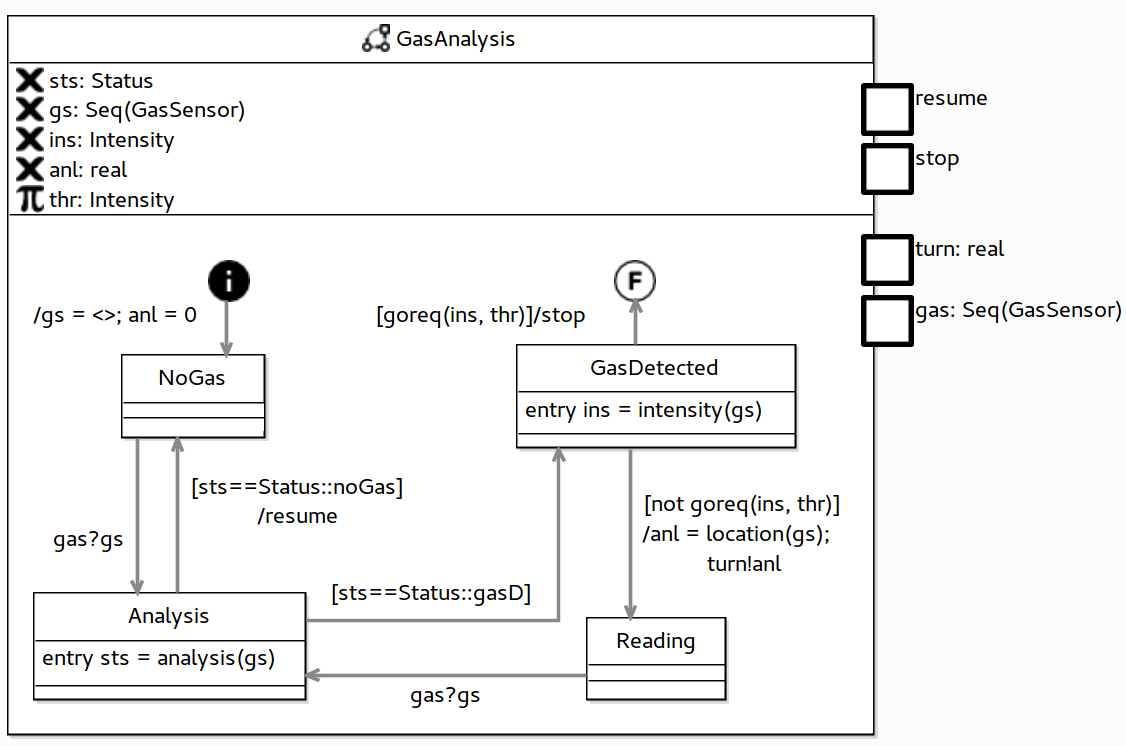}
  \end{center}
  \vspace{-4.2ex}
  \caption{\textsf{GasAnalysis} state machine in RoboChart}
  \label{fig:gasanalysis}
  \vspace{-3ex}
\end{figure}

The interface consists of four events. The event $gas$ is used to receive sensor readings, and $turn$ is used to
communicate a change of direction. The remaining events, $resume$ and $stop$ carry no data, and are used to communicate
that the robot should resume its searching activities, or stop. The state machine uses four state variables: $sts$ to
store the gas analysis status, $gs$ to store the present reading, $ins$ to store the reading intensity, and $anl$ to
store the angle the robot is pointing. It also has a constant $thr$ for the gas intensity threshold. RoboChart provides
basic types for variables and constants, including integers, real numbers, sets, and sequences (\textit{Seq(t)}). The
user can also define additional types, that can be records, enumerations, or entirely abstract. For example, the type
\textit{Status} is an enumerated type with constructors \textit{gasD} and \textit{noGas}.

The behaviour is described by 6 nodes, including an initial node (\textit{i}); a final node (\textit{F}); and four
states: \textit{NoGas}, \textit{Analysis}, \textit{GasDetected}, and \textit{Reading}. The transitions are decorated
with expressions of the form \textit{trigger}[\textit{condition}]/\textit{statement}. When the event \textit{trigger}
happens and the guard \textit{condition} is true, then \textit{statement} is executed, before transitioning to the next
state. All three parts can optionally be omitted. RoboChart also permits states to have entry, during, and exit
actions. In our example, both \textit{Analysis} and \textit{GasDetected} have entry actions.

The syntax of actions is given below, which assumes a context where event and state variable identifiers have been
specified using the nonterminal \textit{ID}.

\begin{definition}[Action Language Syntax]
  \begin{center}
    \begin{tabular}{rl}
      Action ~:=~~& Event | \ckey{skip} | ID := Expr | Action ; Action | \\[.5ex]
                  & \ckey{if}~Expr~\ckey{then}~Action~\ckey{else}~Action~\ckey{end} \\[1ex]
      Event  ~:=~~& ID | ID\,\ckey{?}\,ID | ID\,\ckey{!}\,Expr
    \end{tabular}
  \end{center}    
\end{definition}

\noindent An action is either an event, a \ckey{skip}, an assignment, a sequential composition, or a conditional. An
event is either a simple synchronisation on some identified event $e$, an input communication ($e?x$) that populates a
variable $x$, or an output event ($e!v$). We omit actions related to time and operations for now.

Modelling with RoboChart is supported by the Eclipse-based
RoboTool\footnote{\url{https://www.cs.york.ac.uk/circus/RoboCalc/robotool/}}, from which Figure~\ref{fig:gasanalysis}
was captured. RoboTool automates verification via model checking using FDR4, and its extension to incorporate the
verification approach presented here is ongoing work.

\subsection{Unifying Theories of Programming}

UTP~\cite{Hoare&98,Cavalcanti&06} is a framework for the formalisation of computational semantic domains that are used
to give denotational semantics to a variety of programming and modelling languages. It employs alphabetised binary
relations to model programs as predicates relating the initial values of variables ($x$) to their later values
($x'$). UTP divides variables into two classes: (1) program variables, that model data, and (2) observational variables,
that encode additional semantic structure. For example, $clock : \nat$ is a variable to record the passage of
time. Unlike a program variable, it makes no sense to assign values to $clock$, as this would model arbitrary time
travel. Therefore, observational variables are constrained using healthiness conditions, which are encoded as idempotent
functions on predicates. For example, application of $\healthy{HT}(P) \defs (P \land clock \le clock')$ results in a
healthy predicate that specifies there is no reverse time travel.

The observational variables and healthiness conditions give rise to a subset of the alphabetised relations called a
\emph{UTP theory}, which can be used to justify the fundamental theorems of a computational paradigm. A UTP theory is
the set of fixed points of the healthiness condition: $\theoryset{H} \defs \{P | \healthy{H}(P) = P\}$. A set of
signature operators can then be defined, under which the theory's healthiness conditions are closed, and are thus
guaranteed to construct programs that satisfy these theorems. UTP theories allow us to model a variety of paradigms
beyond simple imperative programs, such as concurrency~\cite{Hoare&98,Oliveira&09}, real-time~\cite{Sherif2010}, object
orientation~\cite{SCS06}, hybrid~\cite{Foster16b,Foster17b}, and probabilistic systems~\cite{Bresciani2012}.

The use of relational calculus means that the UTP lends itself to automated program verification using refinement
$S \refinedby P$: program $P$ satisfies specification $S$. Since both $S$ and $P$ are specified in formal logic, and
refinement equates to reverse implication, we can utilise interactive and automated theorem proving technology for
verification. This allows application of tools like Isabelle/HOL to program verification, which is the goal of our tool,
Isabelle/UTP~\cite{Foster16a}.

\subsection{Isabelle/HOL and Isabelle/UTP}

Isabelle/HOL~\cite{Isabelle} consists of the Pure meta-logic, and the HOL object logic. Pure provides a term language,
polymorphic type system, syntax translation framework for extensible parsing and pretty printing, and an inference
engine. The jEdit-based IDE allows \LaTeX-like term rendering using Unicode. An Isabelle theory consists of type
declarations, definitions, and theorems, which are usually proved by composition of existing theorems. Theorems have the
form of $\llbracket P_1; \cdots; P_n \rrbracket \Longrightarrow Q$, where $P_i$ is an assumption, and $Q$ is the
conclusion. The simplifier tactic, \textsf{simp}, rewrites terms using theorems of the form
$f(x_1 \cdots x_n) \equiv y$.

HOL implements an ML-like functional programming language founded on an axiomatic set theory similar to ZFC. HOL is
purely definitional: mathematical libraries are constructed purely by application of the foundational axioms, which
provides a highly principled framework. HOL provides inductive datatypes, recursive functions, and records. Several
basic types are provided, including sets, functions, numbers, and lists. Parametric types are written by precomposing
the type name, $\tau$, with the type variables $[a_1, \cdots, a_n]\tau$, for example $[nat]list$\footnote{The square
  brackets are not used in Isabelle; we add them for readability.}.

Isabelle/UTP~\cite{Feliachi2010,Foster16a} is a semantic embedding of UTP into HOL, including a formalisation of the
relational calculus, fundamental laws, proof tactics, and facilities for UTP theory engineering. The relational calculus
is constructed such that properties can be recast as HOL predicates, and then automated tactics, such as \textsf{auto},
and \textsf{sledgehammer}~\cite{Blanchette2011}, can be applied. This strategy is employed by our workhorse tactic,
\textsf{rel-auto}, which automates proof of relational conjectures.


Proof automation is facilitated by encoding variables as lenses~\cite{Foster16a}. A lens $x :: \tau \lto \alpha$
characterises a $\tau$-shaped region of the type $\alpha$ using two functions: $\lget_x :: \alpha \to \tau$ and
$\lput_x :: \alpha \to \tau \to \alpha$, that query and update the region, respectively. Intuitively, $x$ is a variable
of type $\tau$ within the alphabet type $\alpha$. Alphabet types can be encoded using the
$\ckey{alphabet}~r = f_1 :: \tau_1 \cdots f_n :: \tau_n$ command, that constructs a new record type $r$ with $n$ fields,
and a lens for each field. Lenses can be independent, meaning they cover disjoint regions, written $x \lindep y$, or
contained within another, written $x \lsubseteq y$. These allow us to express meta-logical style properties without
actually needing a meta-logic~\cite{Foster16a}.

The core UTP types include predicates $[\alpha]\upred$, and (homogeneous) relations $[\alpha]\uhrel$. Operators are
denoted using lenses and lifted HOL functions. An important operator is substitution, $\substapp{\sigma}{P}$, which
applies a state update function $\sigma :: \alpha \to \alpha$ to an expression, and replaces variables in a similar way
to syntactic substitution. Substitutions can be built using lens updates $\sigma(x \mapsto v)$, for
$x :: \tau \lto \alpha$ and $v :: [\tau, \alpha]\uexpr$, and we use the notation
$\lupdate{x_1 \mapsto v_1, \cdots, x_n \mapsto v_n}$ for a substitution in $n$ variables. Substitution theorems can be
applied with the simplifier to perform symbolic evaluation of an expression.

All the theorems and results that we present in this paper have been mechanically validated in Isabelle/UTP, and the
proofs can be found in our repository\footnote{\url{https://github.com/isabelle-utp/utp-main/tree/master/robochart/untimed}}.



\subsection{Stateful-Failure Reactive Designs}

RoboChart is a reactive language, where controllers exchange events with one another and the robotic platform or the
environment. Reactive programs can make decisions both internally, based on the evaluation of their own state, and
externally, by offering several events. Consequently, they pause at particular quiescent points during execution, when
awaiting a communication. Unlike sequential programs, they need not terminate but may run indefinitely.

The UTP theory of stateful-failure reactive designs~\cite{Foster18a,Foster-SFRD-UTP} exists to give denotational
semantics to reactive programming languages, such as CSP~\cite{Hoare&98}, \Circus~\cite{Oliveira&09}, and
rCOS~\cite{Zhan2008}. It is a relational version of the stable failures-divergences semantic model, as originally
defined in the UTP book~\cite{Hoare&98,Cavalcanti&06} using event traces and refusal sets, but extended with state
variables. Its healthiness condition, $\healthy{NCSP}$, which we previously mechanised~\cite{Foster18a}, characterises
relations that extend the trace, update variables, and refuse certain events in quiescent phases. The signature includes
unbounded nondeterministic choice ($\bigsqcap_{i \in I} P(i)$), conditional ($\conditional{P}{b}{Q}$), and sequential
composition ($P \relsemi Q$).  $\theoryset{NCSP}$ forms a complete lattice under $\refinedby$, with top element
$\Miracle$ and bottom $\Chaos$, and also a Kleene algebra~\cite{Foster18a}, which allows reasoning about iterative
reactive programs.

The signature also contains several specialised operators. Event action, $\Do{e}$, describes the execution of an event
expression $e$, that ranges over state variables. When activated, it waits for $e$ to occur, and then it
terminates. Generalised assignment ($\assigns{\sigma}$) uses a substitution $\sigma$ to update the state, following
Back~\cite{Back1998}. Basic assignment can be defined as $(x := v) \defs \assigns{\lupdate{x \mapsto v}}$, and a unit as
$\Skip \defs \assigns{id}$.  External choice, $\Extchoice i\!\in\!A @ P(i)$ indexed by set $A$, as in CSP, permits one
of the branches to resolve either by an event, or by termination. A binary choice $P \extchoice Q$ is denoted by
$\Extchoice X \in \{P, Q\} @ X$. A guard, $b \guard P$, executes $P$ when $b$ is true, and is otherwise equivalent to
$\Stop$, the deadlocked action. These operators obey several algebraic laws~\cite{Foster18a}; a small selection is
below.

\begin{theorem} \label{thm:ncsplaws}
  If $P$ is \healthy{NCSP}-healthy, then the following identities hold:

  \vspace{-3.5ex}
  \begin{align}
    \Miracle \relsemi P =~& \Miracle \label{thm:ncspmiracle} \\
    \assigns{\sigma} \relsemi P =~& \substapp{\sigma}{P} \label{thm:ncspasn} \\
    (\Do{a}\!\extchoice\!\Do{b})\!\relsemi\!P =~& \Do{a}\!\relsemi\!P \extchoice \Do{b}\!\relsemi\!P \label{thm:ncspdist}
  \end{align}
\end{theorem}


\noindent\eqref{thm:ncspmiracle} states that miracle is a left annihilator for sequential composition. \eqref{thm:ncspasn} allows
us to push an assignment into a successor program by inserting a substitution. \eqref{thm:ncspdist} allows us to left
distribute through an external choice of events.

Our theory supports specifications using reactive contracts: $\rc{\!P\!}{\!Q\!}{\!R\!}$. It consists of three relations
over the state variables, trace variable ($\trace$), and refusal set variable ($\refu$). $P$ characterises assumptions
of the initial state and trace, $Q$ characterises quiescent behaviours, and $R$ characterises terminating
behaviours. Our previous result~\cite{Foster18a} shows that any reactive program can be denoted using a reactive
contract, which can be calculated by equational laws. This enables a verification strategy that checks refinements
between a specification and implementation contract, and has been implemented in a tactic called
\textsf{rdes-refine}~\cite{Foster17c}, that can be used to check for deadlock~\cite{Foster18a}, and which we employ in
this paper.

\section{Foundations for State Diagrams}
\label{sec:foundations}


In this section we extend the theory of reactive designs with constructs necessary to denote state machines, and prove
several theorems, notably an induction law for iterative programs. Although these programming constructs are rather
standard, we consider their semantics in the reactive programming paradigm, rather than in the standard sequential
programming setting. It is a pleasing aspect of our approach that standard laws hold in this much richer context.


State machines describe how to transition from one node to another. The main construct we use to denote them is a
reactive version of Dijkstra's guarded iteration statement~\cite{Dijkstra75} $\doiter{i\!\in\!I}{b(i)}{P(i)}$, which
repeatedly selects an indexed statement $P(i)$, based on whether its respective guard $b(i)$ is true. $I$ is an index
set, which when finite gives rise to the more programmatic form of
$\ckey{do} ~ b_1 \gcmd P_1 | \cdots | b_n \gcmd P_n ~ \ckey{od}$. We begin by defining Dijkstra's alternation construct,
$\alternate{i\!\in\!I}{b(i)}{P(i)}$~\cite{Dijkstra75}, which underlies iteration.
\begin{definition}[Guarded Commands, Assumptions, and Alternation]
  \begin{align*}
    b \gcmd P ~\defs~~& \conditional{P}{b}{\Miracle} \\
    [b] ~\defs~~& b \gcmd \Skip \\
    \alternate{i\!\in\!I}{b(i)}{P(i)} ~\defs~~& \textstyle\left(\bigsqcap_{i \in I}~b(i) \gcmd P(i)\right) \sqcap \left(\left(\neg \bigvee_{i \in I}~b(i)\right) \gcmd \Chaos\right)
  \end{align*}
\end{definition}
\noindent $b \gcmd P$ is a ``naked'' guarded command~\cite{Morgan92}. Its behaviour is $P$ when $b$ is true, and
miraculous otherwise, meaning it is impossible to execute. By Theorem~\ref{thm:ncsplaws}, $\Miracle$ is a left
annihilator for sequential composition, and so any following behaviour is excluded when $b$ is false. An assumption
$[b]$ guards $\Skip$ with $b$, and thus holds all variables constant when $b$ is true, and is otherwise
miraculous. $\theoryset{NCSP}$ is closed under both these operators since they are defined only in terms of healthy
elements $\conditional{}{\cdot}{}$, $\Miracle$, and $\Skip$.

Alternation is a nondeterministic choice of guarded commands. When $b(i)$ is true for $i \in I$, $P(i)$ can be
executed. Any command which has $b(i)$ false evaluates to $\Miracle$ and thus is eliminated. If no $b(i)$ is true, then
its behaviour is $\Chaos$. If multiple $b(i)$ are true then one of the corresponding $P(i)$ is nondeterministically
selected. $\theoryset{NCSP}$ is closed under alternation since it comprises only healthy elements. From this definition
we can prove a number of characteristic laws.
\begin{theorem} If, $\forall i @ P(i) \is \healthy{NCSP}$, then the following identities hold:
  \begin{align}
  \alternate{i \in \emptyset}{b(i)}{P(i)}                         ~=~& \Chaos \label{law:altempty} \\
  \alternate{i \in \{k\}}{b(i)}{P(i)}                             ~=~& \conditional{P(k)}{b(k)}{\Chaos} \label{law:altsingle} \\
  \textstyle\casm{\bigvee_{i \in I} \, b(i)} \relsemi \alternate{i \in I}{b(i)}{P(i)} ~=~& \textstyle\left(\bigsqcap_{i \in I}~b(i) \gcmd P(i)\right) \label{law:altasm}
  \end{align}
\end{theorem}
\noindent In words, \eqref{law:altempty} shows that alternation over an empty set presents no options, and so is
equivalent to $\Chaos$; \eqref{law:altsingle} shows that a singleton alternation can be rewritten as a binary
conditional; \eqref{law:altasm} shows that, if we assume that one of its branches is true, then an alternation
degenerates to a nondeterministic choice.

We now define guarded iteration as the iteration of the corresponding alternation whilst at least one of the guards
remains true.

\begin{definition}[Guarded Iteration]
  $$\textstyle\doiter{i\!\in\!I}{b(i)}{P(i)} ~\defs~~ \while{\left(\bigvee_{i \in I}~b(i)\right)}{\left(\alternate{i\!\in\!I}{b(i)}{P(i)}\right)}$$
\end{definition}

\noindent We use the reactive while loop ($\while{b}{P}$) to encode the operator, and can thus utilise our previous
results~\cite{Foster18a} to reason about it. In keeping with the reactive programming paradigm, this while loop can
pause during execution to await interaction, and it also need not terminate. However, in order to ensure that the
underlying fixed point can be calculated, we assume that for all $i \in I$, $P(i)$ is productive~\cite{Foster17b}: that
is, it produces at least one event whenever it terminates. This ensures that divergence caused by an infinite loop is
avoided. Iteration is closed under $\theoryset{NCSP}$, since the while loop and alternation both are.

We can now prove the following fundamental refinement law for iteration.

\begin{theorem}[Iteration Induction] \label{thm:iterintro} If, $\forall i @ P(i) \is \healthy{NCSP}$, then:
  $$\begin{array}{c}
      \begin{array}{lll}
        \forall i \in A @ P(i) \text{ is Productive} &\quad& S \refinedby I \relsemi [\bigwedge_{i \in A} \, (\neg b(i))] \\[1ex]
      \forall i \in A @ S \refinedby I \relsemi [b(i)] \relsemi P(i) &\quad& \forall i \in A @ S \refinedby S \relsemi [b(i)] \relsemi P(i)
      \end{array}
      \\[1ex] \\[-1.0em] \hline \\[-1.0em]  
      S \refinedby I \relsemi \doiter{i \in A}{b(i)}{P(i)}
  \end{array}$$
\end{theorem}

\noindent The law states the provisos under which an iteration, with initialiser $I$, preserves invariant $S$. These
are: (1) every branch is productive; (2) if $I$ causes the iteration to exit immediately then $S$ is satisfied; (3) for
any $i \in A$ if $I$ holds initially, $b(i)$ is true, and $P(i)$ executes, then $S$ is satisfied (base case); and (4)
for any $i \in A$ if $b(i)$ is true, and $P(i)$ executes, then $S$ is satisfied (inductive case). This law forms the
basis for our verification strategy.


\section{Mechanised Reactive Programs}
\label{sec:actions}
In this section we turn our reactive design theory into an Isabelle/HOL type, so that we can use the type system to
ensure well-formedness of reactive programs, which supports our verification strategy. The type allows efficient proof
and use of the simplifier to perform rewriting and also symbolic evaluation so that assignments can be pushed forward
and substitutions applied. We use it to encode both state machine actions in \S\ref{sec:statsem}, and the dynamic
semantics in \S\ref{sec:dynsem}. We first describe a general result for mechanising programs, apply it to reactive
programs, and also introduce a novel operator to express frame extension.


In UTP, all programs are unified by encoding them in the alphabetised relational calculus. Programs in different
languages of various paradigms therefore have a common mathematical form, and can be both compared and semantically
integrated. This idea is retained in Isabelle/UTP by having all programs occupy the type $[\alpha]\uhrel$, with a
suitably specialised alphabet type $\alpha$~\cite{Feliachi2010}.

In Isabelle/UTP, we characterise a theory by (1) an alphabet type $\mathcal{T}$, which may be parametric; and (2) a
healthiness function, $\healthy{H} :: [\mathcal{T}]\uhrel \to [\mathcal{T}]\uhrel$. The theory signature consists of
operators with the form $f_i :: ([\mathcal{T}]\uhrel)^k \to [\mathcal{T}]\uhrel$, each of which is accompanied by a
proven closure theorem
$$\textit{f-H-closed}: \llbracket P_1 \is \healthy{H}; \cdots ; P_k \is \healthy{H} \rrbracket \Longrightarrow f(P_1, \cdots, P_k) \is \healthy{H}$$
which ensures that the operator constructs healthy elements, provided its parameters are all healthy. For example, the
reactive design theory has a theorem $\llbracket P \is \healthy{NCSP}; Q \is \healthy{NCSP} \rrbracket \Longrightarrow (P \relsemi Q) \is \healthy{NCSP}$,
which demonstrates that sequential composition is in the signature. Theories also typically have algebraic laws, like
those in Theorem~\ref{thm:ncsplaws}, which can be applied to reasoning about programs and thence to produce
verification tools~\cite{Foster17c,Foster18a}.

This approach has several advantages for theory engineering~\cite{Hoare&98,Cavalcanti&06}. There is a unified notion of
refinement that can be applied across semantic domains. Operators like nondeterminsitic choice ($\sqcap$) and sequential
composition ($\relsemi$) can occupy several theories, which facilitates generality and semantic integration. General
algebraic laws can be proved, and then directly reused in more specialised UTP theories. The UTP approach means that
theories can be both combined and extended for a wide variety of computational paradigms and languages.

However, there is a practical downside, which is that the programming theorems, such as those in
Theorem~\ref{thm:ncsplaws}, require healthiness of the constituent parameters, and therefore it is necessary to
first invoke the closure theorems. In the context of verification, constantly proving closure can be very inefficient,
particularly for larger programs. This is because Isabelle's simplifier works best when invoked with pure equations
$f(x_1, \cdots, x_n) \equiv y$ with minimal provisos. 


Our solution uses the Isabelle type system to shoulder the burden of closure proof. We use the \ckey{typedef} mechanism,
which creates a new type $T$ from a non-empty subset $A :: \power(U)$ of existing type $U$. For a UTP theory, we create
a type with $A = \theoryset{H}$, which is a subset of the UTP relations. This then allows optimised proof for a
particular UTP theory, but at the cost of generality and semantic extensibility which are more suited to the UTP
relational domain.

In order to obtain the signature for the new type, we utilise the lifting package~\cite{Huffman13}, whose objective is
to define operators on $T$ in terms of operators on $U$, provided that $A$ is closed under each operator. Specifically,
if $f$ is a signature operator in $k$ arguments, then we can create a lifted operator $\widehat{f} :: T^k \to T$ using
Isabelle's \ckey{lift-definition} command~\cite{Huffman13}. This raises a proof obligation that
$f \in \theoryset{H}^k \to \theoryset{H}$, which can be discharged by the corresponding closure theorem. Programs
constructed from the lifted operators are well-formed by construction.

Finally, to lift the algebraic theorems for each lifted operator $\widehat{f}$, we use the \textit{transfer}
tactic~\cite{Huffman13}. It allows us to prove theorems like
$\widehat{f}(P_1, \cdots, P_k) = \widehat{g}(P_1, \cdots, P_k)$, where $P_i :: T$ is a free variable, by converting it
to a theorem of the form
$\llbracket Q_1 \is \healthy{H}; \cdots; Q_k \is \healthy{H} \rrbracket \Longrightarrow f(Q_1, \cdots, Q_k) = g(Q_1,
\cdots, Q_k)$.
This means the closure properties of each parameter $Q_i$ can be utilised in disharging provisos of the corresponding
UTP theorems, but the lifted theorems do not require them. We will now use this technique for our reactive program type.



The reactive designs alphabet is $[s, e]\textit{st-csp}$, for state space $s$ and event type
$e$. \ckey{NCSP}~\cite{Foster18a}, of type $[[s, e]\textit{st-csp}]\uhrel \to [[s, e]\textit{st-csp}]\uhrel$,
characterises the theory. We use it to define the reactive program type, $[s,e]\Action$ and lift each theory operator
from \S\ref{sec:prelim} and \S\ref{sec:foundations}. For example, guard is a function $(b\!\guard\! P) :: [s,e]\Action$,
for $b :: [s]\upred$ and $P :: [s,e]\Action$. For the action language, we define basic events $e \defs \Do{e}$, send
$e!v \defs \Do{e.v}$, and receive $e?x \defs \Extchoice v @ \Do{e.v} \relsemi x\!:=\!v$. From these lifted definitions,
and using the \textit{transfer} tactic, all the laws in Theorems~\ref{thm:ncsplaws} and \ref{thm:iterintro} can be
recast for the new operators, but without closure conditions. We then prove substitution laws for
$\substapp{\sigma}{P}$, where $\sigma :: s \to s$ and $P :: [s,e]\Action$, which can be used for symbolic evaluation.
\begin{theorem}[Symbolic Evaluation Laws] \label{thm:symeval}

  \vspace{-2ex}

  \begin{minipage}{0.35\textwidth}
    \begin{align*}
    \substapp{\sigma}{[b]} =~& [\substapp{\sigma}{b}] \relsemi \assigns{\sigma} \\
    \substapp{\sigma}{(P \relsemi Q)} =~& (\substapp{\sigma}{P}) \relsemi Q \\
    \substapp{\sigma}{\assigns{\rho}} =~& \assigns{\rho \circ \sigma}
    \end{align*}
  \end{minipage}
  \begin{minipage}{0.65\textwidth}
    \begin{align*}
    \substapp{\sigma}{(P \extchoice Q)} =~& (\substapp{\sigma}{P}) \extchoice (\substapp{\sigma}{Q}) \\
    \substapp{\sigma}{(b \guard P)} =~& (\substapp{\sigma}{b}) \guard (\substapp{\sigma}{P}) \\
    \substapp{\sigma}{e!v} =~& e!(\substapp{\sigma}{v}) \relsemi \assigns{\sigma}
    \end{align*}
  \end{minipage}
\end{theorem}
\noindent These laws show how substitution applies and distributes through the operators. In combination with the
assignment law of Theorem~\ref{thm:ncsplaws}(\ref{thm:ncspasn}), they can be used to apply state updates. For example,
one can automatically prove that
$$(x := 2 \relsemi y := (3 * x) \relsemi e!(x + y)) ~~ = ~~ (e!8 \relsemi \assigns{x \mapsto 2, y \mapsto 6})$$ since we can combine the assignments and push them through the send event.


To denote state machines, we need a special variable ($actv$) to record the currently active node. This is semantic
machinery, and no action is permitted access to it. We impose this constraint via frame extension:
$\frext{a}{P} :: [s_1, e]Action$, for $a :: s_2 \lto s_1$ and $P :: [s_2, e]Action$, that extends the alphabet of $P$. It
is similar to a frame in refinement calculus~\cite{Morgan92}, which prevents modification of variables, but also uses
the type system to statically prevent access to them. Lens $a$ identifies a subregion $\alpha$ of the larger alphabet
$\beta$, that $P$ acts upon. Intuitively, $\alpha$ is the set of state machine variables, and $\beta$ this set extended
with $actv$. $P$ can only modify variables within $\alpha$, and others are held constant. We prove laws for this
operator, which are also be used in calculating the semantics.

\begin{theorem}[Frame Extension Laws]
  \vspace{-3ex}

  $$\frext{a}{P\!\relsemi\!Q} =~ \frext{a}{P}\!\relsemi\!\frext{a}{Q} \quad \frext{a}{e?x} =~ e?(\lns{a}{x}) \quad \frext{a}{x := v} =~ \lns{a}{x} := v$$

\end{theorem}

\noindent Frame extension distributes through sequential composition. For operators like event receive and assignment,
the variable is extended by the lens $a$, which is like a namespace operator ($\lns{a}{x}$). Specifically, it manipulates
the region characterised by $x$ within the region of $a$. This completes the mechanised reactive language.

\section{Static Semantics}
\label{sec:statsem}
In this section we formalise a state machine meta-model in Isabelle/HOL, which describes the variables, transitions, and
nodes. The meta-model, presented below, is based on the untimed subset of RoboChart, but note that our use of UTP
ensures that our work is extensible to more advanced semantic domains~\cite{Sherif2010,Bresciani2012,Foster17b}. For now
we omit constructs concerned with interfaces, operations, shared variables, during actions, and hierarchy, and focus on
basic machines.

\begin{definition}[State Machine Meta-Model]
  \vspace{1ex}

  \noindent
  \begin{tabular}{rl}
    StMach :=& \ckey{statemachine}~ID~ \\
                   & ~~ \ckey{vars} NameDecl* \ckey{events} NameDecl*  \ckey{states} NodeDecl*\\ 
                   & ~~  \ckey{initial} ID \ckey{finals} ID* \ckey{transitions} TransDecl* \\[1ex]
    NameDecl      :=& ID [\ckey{:} Type] \\[1ex]
    NodeDecl     :=& ID \ckey{entry} Action \ckey{exit} Action \\[1ex]
    TransDecl    :=& ID \ckey{from} ID \ckey{to} ID \ckey{trigger} Event \ckey{condition} Expr \ckey{action} Action
  \end{tabular}
\end{definition}

\noindent A state machine is composed of an identifier, variable declarations, event declarations, state declarations,
an initial state identifier, final state identifiers, and transition declarations. Each variable and event consists of a
name and a type. A state declaration consists of an identifier, entry action, and exit action. A transition declaration
consists of an identifier, two state identifiers for the source and target nodes, a trigger event, a condition, and a
body action. Whilst we do not directly consider hierarchy, this can be treated by flattening out substates.

We implement the meta-model syntax using Isabelle's parser, and implement record types $[s, e]\textit{Node}$ and
$[s, e]\textit{Transition}$, that correspond to the \textit{NodeDecl} and \textit{TransDecl} syntactic categories. They
are both parametric over the state-space $s$ and event types $e$. \textit{Node} has fields $nname :: \textit{string}$,
$nentry :: [s, e]\Action$, and $nexit :: [s,e]\Action$, that contain the name, entry action, and exit
action. \textit{Transition} has fields $src :: \textit{string}$, $tgt :: \textit{string}$, $trig :: [s, e]\Action$,
$cond :: [s]\upred$, and $act :: [s, e]Action$, that contain the source and target, the trigger, the condition, and the
body. We then create a record type to represent the state machine.

\begin{definition}[State Machine Record Type]
  $$\ckey{record}~[s, e]\textit{StMach} ~~=~~ \begin{array}{l}
     \begin{array}{ll} init :: ID & finals :: [ID]list \\[.5ex]
        nodes :: [[s, e]Node]list \quad& transs :: [[s, e]Transition]list
      \end{array}
    \end{array}$$
\end{definition}

\noindent It declares four fields for the initial state identifier (\textit{init}), final states identifiers
(\textit{finals}), nodes definitions (\textit{nodes}), and transition definitions (\textit{transs}), and constitutes the
static semantics. Since this corresponds to the meta-model, and to ensure a direct correspondence with the parser, we do
not directly use sets and maps, but only lists in our structure. We will later derive views onto the data structure
above, that build on well-formedness constraints.

Below, we show how syntactic machines are translated to Isabelle definitions.

\begin{definition}[Static Semantics Translation]
$$
  \begin{array}{l}
    \ckey{statemachine}~s \\[1ex]
    ~ \ckey{vars}~x_1 : \tau^v_1 ~\cdots~ x_i : \tau^v_i \\[1ex]
    ~ \ckey{events}~e_1 : \tau^e_1 ~\cdots e_j~ : \tau^e_j \\[1ex]
    ~ \ckey{states}~s_1 \cdots s_k ~ \ckey{initial}~ini \\[1ex]
    ~ \ckey{finals}~f_1 \cdots f_m \\[1ex]
    ~ \ckey{transitions}~t_1 \cdots t_n
  \end{array}
  ~~~\Longrightarrow~~~~~ \begin{array}{l}
  \ckey{alphabet}~\text{s-alpha} = x_1 : \tau^v_1 ~\cdots~ x_i : \tau^v_i \\[1ex]
  \ckey{datatype}~\text{s-ev} = \nulle | e_1~t^e_1 | \cdots | e_j~t^e_j \\[1ex]
  \ckey{definition}~ machine :: [\text{s-alpha}, \text{s-ev}]\text{StMach} \\[1ex]
  \ckey{where}~machine = \begin{array}{l}
                           \llparenthesis init = ini, \\ 
                              ~finals = [f_1 \cdots f_m], \\
                              ~states = [s_1 \cdots s_k], \\
                              ~transs = [t_1 \cdots t_n] \rrparenthesis
                          \end{array} \\ \\[-2.5ex]
  \ckey{definition}~ semantics = \msem{machine}
\end{array}
$$
\end{definition}
%
%
\noindent For each machine, a new alphabet is created, which gives rise to a HOL record type \textit{s-alpha}, and
lenses for each field of the form $t^v_i \lto \textit{s-alph}$. For the events, an algebraic datatype \textit{s-ev} is
created with constructors corresponding to each of them. We create a distinguished event $\nulle$ that will be used in
transitions with explicit trigger and ensures productivity. The overall machine static semantics is then contained in
\textit{machine}. We also define \textit{semantics} that contains the dynamic semantics in terms of the semantic
function $\msem{\cdot}$ that we describe in \S\ref{sec:dynsem}.

Elements of the meta-model are potentially not well-formed, for example specifying an initial state without a
corresponding state declaration, and therefore it is necessary to formalise well-formedness. RoboTool enforces a number
of well-formedness constraints~\cite{Miyazawa2018RoboChart}, and we here formalise the subset needed to ensure the
dynamic semantics given in \S\ref{sec:dynsem} can be generated. We need some derived functions for this, and so we
define $nnames \defs set(map~nname~(nodes))$, which calculates the set of node names, and $fnames$, which calculates the
set of final node names. We can now specify our well-formedness constraints.
\begin{definition} \label{def:wfsm} A state machine is well-formed if it satisfies these constraints:
  \begin{enumerate}
    \item \textnormal{Each node identifier is distinct:} $distinct(map~nname~(nodes))$
    \item \textnormal{The initial identifier is defined:} $init \in nnames$
    \item \textnormal{The initial identifier is not final:} $init \notin fnames$
    \item \textnormal{Every transition's source node is defined and non-final:} \\ $\forall t \in transs @ src(t) \in nnames \setminus fnames$
    \item \textnormal{Every transition's target node is defined:} $\forall t \in transs @ tgt(t) \in nnames$
  \end{enumerate}
\end{definition}
\noindent We have implemented them in Isabelle/HOL, along with a proof tactic called \textit{check-machine} that
discharges them automatically when a generated static semantics is well-formed, and ensure that crucial theorems are
available to the dynamic semantics. In practice, any machine accepted by RoboTool is well-formed, and so this tactic
simply provides a proof of that fact to Isabelle/HOL.

In a well-formed machine every node has a unique identifier. Therefore, using Definition~\ref{def:wfsm}, we construct
two finite partial functions, $\nmap :: ID \ffun [s, e]\textit{Node}$ and
$\tmap :: ID \ffun [s, e]\textit{Transition}~list$, that obtain the node definition and list of transitions associated
with a particular node identifier, respectively, whose domains are both equal to $\textit{nnames}$. We also define
$\textit{ninit} \defs \nmap~\textit{init}$, to be the definition of the initial node, and \textit{inters} to be the set
of nodes that are not final. Using well-formedness we can then prove the following theorems.
\begin{theorem}[Well-formedness Properties]
  \begin{enumerate}
    \item \textnormal{All nodes are identified:} $\forall n \in set(nodes) @ \nmap~(nname(n)) = n$
    \item \textnormal{The initial node is defined:} $ninit \in set(nodes)$
    \item \textnormal{The name of the initial node is correct:} $nname(ninit) = init$
  \end{enumerate}
\end{theorem}
\noindent These theorems allow us to extract the unique node for each identifier, and in particular for the initial
node. Thus, Isabelle/HOL can parse a state machine definition, construct a static semantics for it, and ensure that this
semantics is both well-typed and well-formed. The resulting Isabelle command is illustrated in
Figure~\ref{fig:rcisabelle} that encodes the \textsf{GasAnalysis} state machine of Figure~\ref{fig:gasanalysis}.

\begin{figure}[t]
  \vspace{-2ex}
  \begin{center}
    \includegraphics[width=\textwidth]{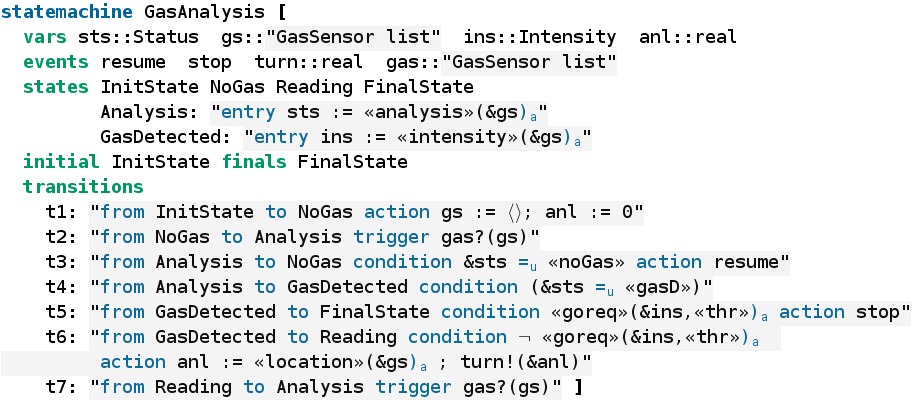}
  \end{center}
  \vspace{-4ex}

  \caption{State machine notation in Isabelle/UTP}
  \label{fig:rcisabelle}

  \vspace{-4ex}
\end{figure}

\section{Dynamic Semantics}
\label{sec:dynsem}
In this section we describe the behaviour of a state machine using the reactive program domain we mechanised in
\S\ref{sec:actions}. The RoboChart reference semantics~\cite{Miyazawa2018RoboChart} represents a state machine as a
parallel composition of CSP processes that represent the individual variables and states. Variable access and state
orchestration are modelled by communications between them. Here, we capture a simpler sequentialised semantics using
guarded iteration, which eases verification. In particular, state variables have a direct semantics, and require no
communication. The relation between these two semantics can be formalised by an automated refinement strategy that
reduces parallel to sequential composition~\cite{Oliveira&09}.

We first define alphabet type $[s]rcst$, parametrised by the state space type $s$, and consisting of lenses
$actv :: ID \lto [s]rcst$ and $\ckey{r} :: s \lto [s]rcst$. The former lens records the currently active state, and the
latter projects the state machine variable space. No action is permitted to refer to $actv$, a constraint that we impose
through the frame extension $\frext{\ckey{r}}{P}$.

We describe the dynamic semantics of a state diagram using three functions.

\begin{definition}[Dynamic Semantics]
  \begin{align*}
    \msem{M} \defs~~& \left(\begin{array}{l}actv := init_M \relsemi \\ \doiter{N\!\in\!set(inters_M)}{actv = nname(N)}{\nsem{M}{N}} 
                      \end{array}\right)
    \\[.5ex]
    \nsem{M}{N} \defs~~ & 
      \begin{array}{l}
        \frext{\ckey{r}}{nentry(N)} \relsemi \left(\Extchoice t\!\in\!\tmap_M(nname(N)) @ \trsem{M,N}{t}\right)
      \end{array} \\[.5ex]
    \trsem{M,N}{t} \defs~~& \frext{\ckey{r}}{cond(t) \guard trig(t) \relsemi nexit(N) \relsemi action(t)} \relsemi actv := tgt(t)
  \end{align*}
\end{definition}

\noindent The function $\msem{\cdot} :: [s, e]StMach \to [[s]rcst, e]Action$ calculates the overall behavioural
semantics. It first sets $actv$ to the initial node identifier, and then enters a do iteration indexed by all non-final
nodes. If a final node is selected, then the iteration terminates. In each iteration, the node $N$ that is named by
$actv$ is selected, and the semantics for it is calculated using $\nsem{M}{N}$.

When in a node, the entry action is first executed using \textit{nentry}, and then an external choice is presented over
all transitions associated with $N$, which are calculated using $\tmap$. The entry and exit actions do not have $actv$
in their alphabet, and therefore we apply frame extensions to them. The semantics of a transition, $\trsem{M,N}{t}$, is
guarded by the transition condition, and awaits the trigger event. Once this occurs, the exit action of $N$ is executed,
followed by the transition action, and finally $actv$ is updated with the target node identifier.

The output of the semantics is an iterative program with one branch for every non-final state. To illustrate, we below
generate the denotational semantics for the \textsf{GasAnalysis} state machine given in Figure~\ref{fig:gasanalysis}.
\begin{example}[\textsf{GasAnalysis} Dynamic Semantics] \label{ex:gadynsem}
  \vspace{-4ex}

  $$\begin{array}{l}
      actv := InitState \relsemi \\[.2ex]
      \ckey{do} \\
      ~~~~\, actv = InitState \then \nulle \relsemi \lns{\ckey{r}}{gs} := \langle\rangle \relsemi \lns{\ckey{r}}{anl} := 0 \relsemi actv := NoGas \\[.2ex]
      ~~| actv = NoGas \then gas?\lns{\ckey{r}}{gs} \relsemi actv := Analysis \\[.2ex]
      ~~| actv = Analysis \then \\ 
          \quad~~ \lns{\ckey{r}}{sts} := analysis(\lns{\ckey{r}}{gs}) \relsemi 
          \left(\begin{array}{l}
             \lns{\ckey{r}}{sts} = noGas \guard \nulle \relsemi resume \relsemi actv := NoGas \\
            \extchoice \lns{\ckey{r}}{sts} = gasD \guard \nulle \relsemi actv := GasDetected 
          \end{array}\right) \\[1ex]
      ~~| actv = GasDetected \then \lns{\ckey{r}}{ins} := intensity(\lns{\ckey{r}}{gs}) \relsemi \\[.5ex]
      \qquad\qquad \left(\begin{array}{l}
               goreq(ins, thr) \guard \nulle \relsemi stop \relsemi actv := FinalState \\
               \extchoice (\neg goreq(ins, thr)) \guard 
                 \nulle \relsemi 
                       \begin{array}{l}
                         \lns{\ckey{r}}{anl} := location(\lns{\ckey{r}}{gs}) \relsemi \\
                         turn!(\lns{\ckey{r}}{anl}) \relsemi actv := Reading
                       \end{array}
                     \end{array}\right) \\[3ex]
      ~~| actv = Reading \then gas?\lns{\ckey{r}}{gs} \relsemi actv := Analysis \\[.5ex]
      \ckey{od}
    \end{array}$$
\end{example}
\noindent In order to yield a more concise definition, we have also applied the action simplification laws given in
\S\ref{sec:actions}. In particular, the frame extensions have all been expanded so that the state variables are
explicitly qualified by lens $\ckey{r}$.

In order to verify such state machines, we need a specialised refinement introduction law. Using our well-formedness
theorem, we can specialise Theorem~\ref{thm:iterintro}.

\begin{theorem} \label{thm:stmref} The semantics of a state machine $M$ refines a reactive invariant specification $S$,
  that is $S \refinedby \msem{M}$, provided that the following conditions hold:
  \begin{enumerate}
    \item $M$ is well-formed according to Definition~\ref{def:wfsm};
    \item the initial node establishes the invariant --- $S \refinedby \nsem{M}{ninit_M}$;
    \item every non-final node preserves $S$ --- $\forall N \in inters_M @ S \refinedby S \relsemi (\nsem{M}{N})$.
  \end{enumerate}
\end{theorem}

\begin{proof}
  By application of Theorem~\ref{thm:iterintro}, and utilising trigger productivity. \qed
\end{proof}


\noindent We now have all the infrastructure needed for verification of state machines, and in the next section we
describe our verification strategy and tool.

\section{Verification Approach}
\label{sec:verify}
In this section we use the collected results presented in the previous sections to define a verification strategy for
state machines, and exemplify its use in verifying deadlock freedom. Our approach utilises Theorem~\ref{thm:stmref} and
our contractual refinement tactic, \textsf{rdes-refine}, to prove that every state of a state machine satisfies a given
invariant, which is specified as a reactive contract. The overall workflow for description and verification of a state
machine is given by the following steps:

\begin{enumerate}
  \item parse, type check, and compile the state machine definition;
  \item check well-formedness (Definition~\ref{def:wfsm}) using the \textit{check-machine} tactic;
  \item calculate denotational semantics, resulting in a reactive program;
  \item perform algebraic simplification and symbolic evaluation (Thms~\ref{thm:ncsplaws},
    \ref{thm:symeval});
  \item apply Theorem~\ref{thm:stmref} to produce sequential refinement proof obligations;
  \item apply \textsf{rdes-refine} to each goal, which may result in residual proof obligations;
  \item attempt to discharge each remaining proof obligation using \textsf{sledgehammer}~\cite{Blanchette2011}.
\end{enumerate}

\noindent Diagrammatic editors, like RoboTool, can be integrated with this by implementing a serialiser for the
underlying meta-model. The workflow can be completely automated since there is no need to enter manual proofs, and the
final proof obligations are discharged by automated theorem provers. If proof fails, Isabelle/HOL has the
\textsf{nitpick}~\cite{Blanchette2011} counterexample generator that can be used for debugging. This means that the
workflow can be hidden behind a graphical tool.

We can use the verification procedure to check deadlock freedom of a state machine using the reactive contract
$\textit{\textsf{dlockf}} \defs \textstyle\rc{\true}{\exists e @ e \notin \refu}{\true}$, an invariant specification
which states that in all quiescent observations, there is always an event that is not being refused. In other words, at
least one event is always enabled; this is the meaning of deadlock freedom. We can use this contract to check the
\textsf{GasAnalysis} state machine. For a sequential machine, deadlock freedom means that it is not possible to enter a
state and then make no further progress. Such a situation can occur if the outgoing transitions can all be disabled
simultaneously if, for example, their guards do not cover all possibilities.

\begin{figure}[t]
  \vspace{-2ex}
  
  \begin{center}
    \includegraphics[width=\textwidth]{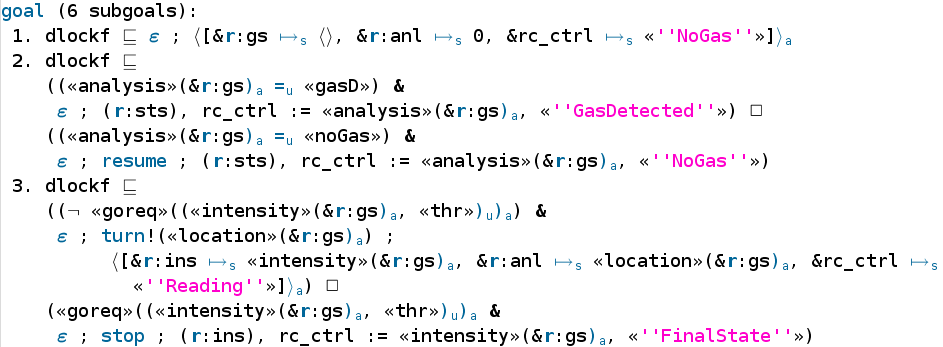}
  \end{center}
  \vspace{-3.5ex}

  \caption{Selection of deadlock freedom proof obligations in Isabelle/UTP}
  \label{fig:rcproofs}

  \vspace{-3.5ex}
\end{figure}

The result of applying the verification procedure up to step 5 is shown in Figure~\ref{fig:rcproofs}. At this stage, the
semantics for each node has been generated, and deadlock freedom refinement conjectures need to be proved. Isabelle
generates 6 subgoals, 3 of which are shown, since it is necessary to demonstrate that the invariant is satisfied by the
initial state and each non-final state. The first goal corresponds to the initial state, where no event occurs and the
variables $gs$ and $anl$, along with $actv$, are all assigned. The second goal corresponds to the \textsf{Analysis}
state. The state body has been further simplified from the form shown in Figure~\ref{ex:gadynsem}, since symbolic
evaluation has pushed the entry action through the transition external choice, and into the two guards. This is also the
case for the third goal, which corresponds to the more complex \textsf{GasDetected} state.

The penultimate step applies the \textsf{rdes-refine} tactic to each of the 6 goals. This produces 3 subgoals for each
goal, a total of 18 first-order proof obligations, and invokes the relational calculus tactic \textsf{rel-auto} on each
of them. The majority are discharged automatically, but in this case three HOL predicate subgoals remain. One of them
relates to the \textsf{Analysis} state, and requires that the constructors $noGas$ and $gasD$ of $Status$ are the only
cases for $sts$. If there was a third case, there would be a deadlock as the outgoing transition guards don't cover
this.

Finally, we execute \textsf{sledgehammer} on each of the three goals, which provides proofs and so completes the
deadlock freedom check. Thus, we have engineered a fully automated deadlock freedom prover for state machines.

  


\section{Conclusions and Related Work}
\label{sec:concl}
In this paper we have presented a verification strategy for state machines in Isabelle/UTP by utilising the theory of
stateful-failure reactive designs, and automated proof facilities. We have extended our UTP theory with the guarded
iteration construct, which is the foundation of sequential state machines, proved a crucial induction law, and adapted
it to an efficient implementation of reactive programs. We have created a static semantics of state machines in
Isabelle/HOL, including well-formedness checks, and a dynamic semantics that generates a reactive program. Finally, we
used this to describe a verification approach that utilises reactive contract refinement and iterative induction.

In future work, we will expand our semantics to handle additional features of RoboChart. Hierarchy, can be handled by
having the $actv$ variable hold a list of nodes, and during actions by implementing a reactive interruption
operator~\cite{McEwan06}. Moreover, we are developing reasoning facilities for parallel composition and hiding to allow
expression of concurrent state machines, which extends our existing work~\cite{Foster17c,Foster18a}. This will greatly
increase verification capabilities for robotic and component-based systems, allow us to handle asynchronous
communication and shared variables, and also to mechanise the CSP reference semantics~\cite{Miyazawa2018RoboChart}.

A challenge that remains is handling assumptions and guarantees between parallel components, but we believe that
abstraction of state machines to invariants, using our results, can make this tractable. We will also explore other
reasoning approaches, such as use of the simplifier to algebraically transform state machines to equivalent forms. Going
further, we emphasise that our UTP theory hierarchy supports more advanced semantic paradigms. We will therefore develop
a mechanised theory of timed reactive designs, based on existing work~\cite{Sherif2010,Foster17b}, and use this to
denote the timing constructs of RoboChart state machines. We are developing a UTP theory of
probability~\cite{Bresciani2012}, and will use it to handle probabilistic junctions. We also have a theory of hybrid
reactive designs~\cite{Foster16b,Foster17b}, which we believe can be used to support hybrid state machines.


In related work, while a number of state machine notations exist, such as UML and Stateflow, to the best of our
knowledge, they provide limited support for formal verification by theorem proving. While formalisations have been
proposed~\cite{Schafer2001,Miyazawa2012}, they typically address a subset of the target notation or focus on model
checking. Other approaches such as~\cite{Foughali2016}, similarly restrict themselves to model checking or other forms
of automatic verification, which have limitations on both the types of systems that can be analysed (mostly finite) and
the kinds of properties that can be checked (schedulability, temporal logic, etc). We differ in that our approach is
extensible, fully automated, and can handle infinite state systems with non-trivial types. Also, our verification laws
have been mechanically validated with respect only to the axioms of Isabelle/HOL.

\vspace{1ex}

\noindent \textbf{Acknowledgements.} This work is funded by the EPSRC projects RoboCalc\footnote{RoboCalc Project:
  \url{https://www.cs.york.ac.uk/circus/RoboCalc/}} (Grant EP/M025756/1) and CyPhyAssure\footnote{CyPhyAssure Project: \url{https://www.cs.york.ac.uk/circus/CyPhyAssure/}} (Grant EP/S001190/1), and the
Royal Academy of Engineering.

\bibliographystyle{splncs}
\bibliography{FACS2018}

\end{document}